\documentclass[a4paper,aps,prl,twocolumn,10pt,superscriptaddress]{revtex4-1}
\usepackage[utf8]{inputenc}
\usepackage[T1]{fontenc}
\usepackage[sc,osf]{mathpazo}
\usepackage{amsmath, amsthm, amssymb}
\usepackage{graphicx}
\usepackage{dcolumn}
\usepackage{dsfont}
\usepackage{bm}
\usepackage{bbm}
\usepackage{hyperref}
\usepackage{tikz}
\usepackage{overpic}

\usepackage{subcaption}

\theoremstyle{plain}
\newtheorem{lemma}{Lemma}
\newtheorem{prop}[lemma]{Proposition}
\newtheorem{theorem}[lemma]{Theorem}

\theoremstyle{definition}

\makeindex

\begin{document}

\title{Conditional Decoupling of Quantum Information}
\date{\today}

\author{Mario Berta}
\affiliation{Department of Computing, Imperial College London, London SW7 2AZ, UK}
\author{Fernando G.~S.~L.~Brand\~{a}o}
\affiliation{Institute for Quantum Information and Matter, California Institute of
Technology, Pasadena, California 91125, USA}

\author{Christian Majenz}
\affiliation{Institute for Language, Logic and Computation, University of Amsterdam, and QuSoft, 1098XG Amsterdam, Netherlands}
\affiliation{Department of Mathematical Sciences, University of Copenhagen, Universitetsparken 5, DK-2100 Copenhagen}

\author{Mark M.~Wilde}
\affiliation{Hearne Institute for Theoretical Physics, Department of Physics and Astronomy, Center for Computation and Technology, Louisiana State University, Baton Rouge, Louisiana 70803, USA}

\begin{abstract}
Insights from quantum information theory show that correlation measures based on quantum entropy are fundamental tools that reveal the entanglement structure of multipartite states. In that spirit, Groisman, Popescu, and Winter [Physical Review A \textbf{72}, 032317 (2005)] showed that the quantum mutual information $I(A;B)$ quantifies the minimal rate of noise needed to erase the correlations in a bipartite state of quantum systems $AB$. Here, we investigate correlations in tripartite systems $ABE$. In particular, we are interested in the minimal rate of noise needed to apply to the systems $AE$ in order to erase the correlations between $A$ and $B$ given the information in system $E$, in such a way that there is only negligible disturbance on the marginal $BE$. We present two such models of conditional decoupling, called deconstruction and conditional erasure cost of tripartite states $ABE$. Our main result is that both are equal to the conditional quantum mutual information $I(A;B|E)$ -- establishing it as an operational measure for tripartite quantum correlations.
\end{abstract}

\maketitle


\paragraph{Introduction.} Landauer's principle states that the amount of work needed for erasing a memory is proportional to the amount of information stored in the memory~\cite{L61}. Motivated by this principle, the correlations of a bipartite quantum state $\rho_{AB}$ shared between two parties Alice and Bob can be quantified by the amount of noise that is required to erase the correlations in $\rho_{AB}$. This erasure cost is closely connected to the thermodynamical cost of erasing the correlations~\cite{GPW05}, which in turn is part of the larger context of the physics of erasure (see, e.g., Refs.~\cite{faist2015minimal,Maruyama2009,plenio2001physics,berut2012experimental,faist2015minimal}). In a model of Groisman, Popescu, and Winter~\cite{GPW05} \footnote{Groisman, Popescu, and Winter discuss various models of how to inject noise into the system; however, ultimately all of them become equivalent.}, Alice is allowed to pick a free ancilla, in the form of an already decoupled state $\theta_{A'}$, and then applies a {\it unitary randomizing channel}
\begin{align}\label{eq:groisman1}
\Lambda_{AA'}(\cdot):=\frac{1}{M}\sum_{i=1}^MU^i_{AA'}\big(\cdot\big)\big(U^i_{AA'}\big)^\dagger\,,
\end{align}
where the noise injected into the system comes from averaging over the unitaries.
The goal is for the resulting state to become close to a product state (or, in other words, {\it decoupled})
\begin{align}\label{eq:groisman2}
F\left(\Lambda_{AA'}(\rho_{AB}\otimes\theta_{A'}),\pi_{A'A}\otimes\rho_B\right)\geq1-\varepsilon\,,
\end{align}
where $\pi_{AA'}$ is a maximally mixed state on a subspace of $AA'$. Here, the action of the channel $\Lambda_{AA'}$ on systems $AA'B$ is understood as $\Lambda_{AA'}\otimes\mathcal{I}_B$, where $\mathcal{I}_B$ denotes the identity channel, and the fidelity between states $\xi$ and $\chi$ is given by $F(\xi,\chi):=[\mathrm{Tr}(\sqrt{\sqrt{\chi}\xi\sqrt{\chi}})]^2$. We note that the use of the ancilla is catalytic in the sense that the system $A'$ has to stay decoupled from $B$ (at least approximately), but potentially makes the erasure process more efficient~\cite{MBDRC16}. The main result of Groisman, Popescu, and Winter~\cite[Thm.~1]{GPW05} is that the minimal rate of unitaries needed in the limit of many copies $\rho_{AB}^{\otimes n}$ and vanishing error $\varepsilon\to0$ is given by the {\it quantum mutual information (QMI)}
\begin{align*}
\frac{1}{n}\log M\to I(A;B)_{\rho}:=H(A)_{\rho}+H(B)_{\rho}-H(AB)_{\rho}\,,
\end{align*}
with the quantum entropy of a state $\eta_{X}$ on system $X$ given by $H(X)_{\eta}:=-\operatorname{Tr}\big[\eta_{X}\log\eta_{X}\big]$. Thus, we can conclude that the QMI is equal to the amount of noise needed for {\it correlation destruction} between systems $A$ and $B$. This result gives information-theoretic justification for the diverse use of the QMI as a correlation measure in quantum physics. For instance, it is a stepping stone in a quantitative understanding of \textit{decoupling}, a central concept both in quantum information theory and in physics in general, with implications ranging from the black-hole information paradox \cite{HayPre07,braunstein-pati,braunstein-zyczkowski} to  area laws in quantum many-body systems \cite{Brandao12}. 


\paragraph{Conditional measures of correlations.} Here, we aim to quantify the correlations in a tripartite quantum state $\rho_{ABE}$. A measure that is (informally) understood as quantifying the correlations between $A$ and $B$ from the perspective of system $E$ is the {\it conditional quantum mutual information (CQMI)}
\begin{align}\label{eq:def_cqmi}
I(A;B|E)_{\rho}:=I(AE;B)_{\rho}-I(E;B)_{\rho}\,.
\end{align}
The CQMI is always non-negative $I(A;B|E)_{\rho}\geq0$, an entropy inequality known as strong sub-additivity~\cite{PhysRevLett.30.434}. The mentioned informal interpretation of the CQMI can be made precise, as it characterizes the resource requirements of the task of quantum state redistribution \cite{DY08} and plays an important role in hypothesis testing of conditional correlations \cite{Tomamichel2018,PhysRevA.94.022310,Berta2017}. The conditional mutual information is also an essential quantity in various areas of physics such as condensed matter physics \cite{Zeng2015,Kim2012}, high energy physics \cite{Pastawski2017,Czech2015}, thermodynamics \cite{Mahajan2016}, and complex and neuronal systems \cite{Bettencourt2008}. The CQMI is closely related to another conditional measure of correlations~\cite{Pet86}, the {\it fidelity of recovery} (FoR)~\cite{SW14}
\begin{align*}
F(A;B|E)_{\rho}:=\sup_{\mathcal{R}_{E\rightarrow AE}}F\big(\rho_{ABE},\mathcal{R}_{E\rightarrow AE}(\rho_{BE})\big)\,,
\end{align*}
where the supremum is with respect to all recovery channels $\mathcal{R}_{E\rightarrow AE}$. The connection of the FoR to the CQMI was only understood very recently in a series of works refining our understanding of multipartite quantum correlations, which began with \cite[Thm.~5.1]{FR14}
\begin{align}\label{eq:fawzi_renner}
I(A;B|E)_{\rho}\geq-\log F(A;B|E)_{\rho}\,.
\end{align}
This shows that the CQMI is a witness to quantum Markovianity: if it is small, then we can understand the correlations between $A$ and $B$ as being mediated by system $E$ via the local recovery channel $\mathcal{R}_{E\rightarrow AE}$. In analogy to the QMI and as a refinement thereof, the CQMI is the basis of various correlation measures in quantum physics. For example, it is a key concept in condensed matter physics, as the CQMI of three regions with a non-trivial topology gives the topological entanglement entropy of the system~\cite{LW06,PK06}. Also in high-energy physics, it has emerged as a important tool to understand the irreversibility of renormalization flow~\cite{L17}.


\paragraph{Deconstruction of quantum correlations.}

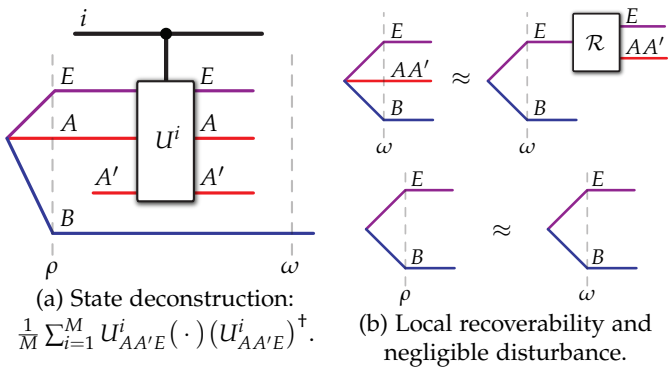
\begin{figure}
\begin{subfigure}{0.25\textwidth}
\begin{overpic}[width=\textwidth]{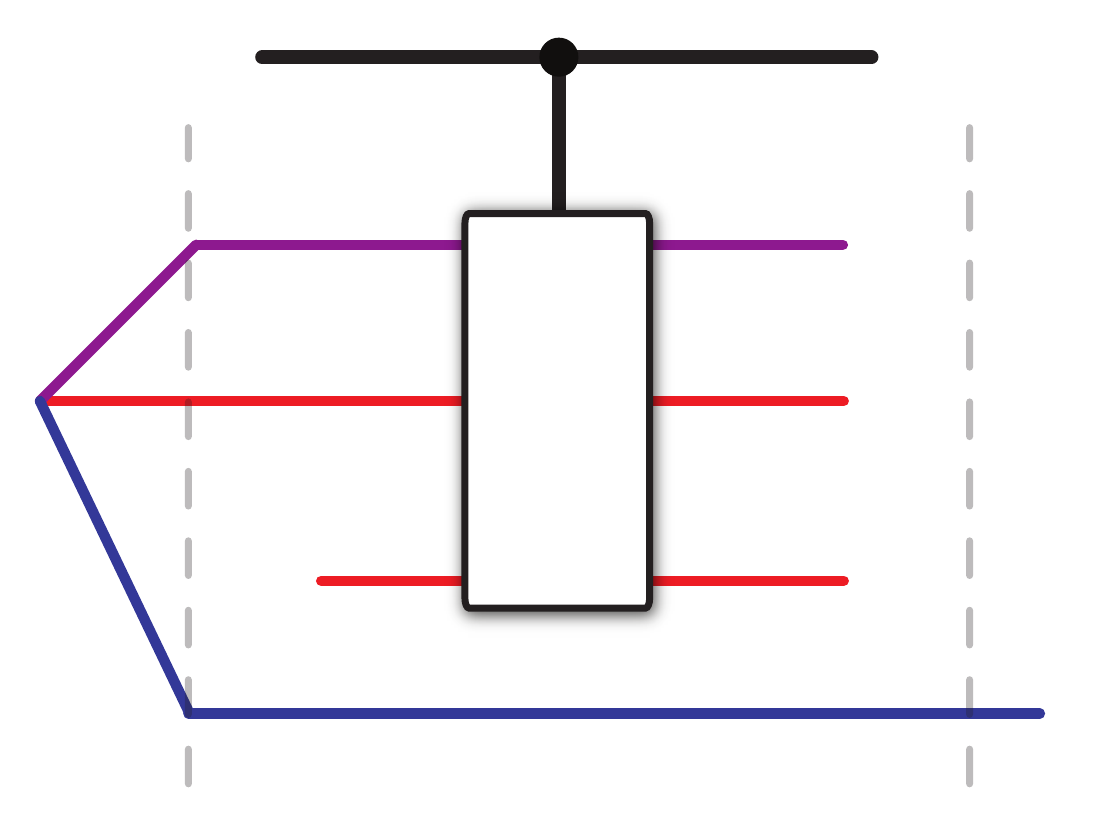}
  \put(25,72){$i$}
  \put(19,55){$E$}
  \put(19,41){$A$}
  \put(29,25){$A'$}
  \put(19,13){$B$}
  \put(47,36){$U^i$}
  \put(60,55){$E$}
  \put(60,41){$A$}
  \put(60,25){$A'$}
  \put(15,-1){$\rho$}
  \put(84,-1){$\omega$}
\end{overpic}
\caption{State deconstruction: $\frac{1}{M}\sum_{i=1}^MU^i_{AA'E}\big(\cdot\big)\big(U^i_{AA'E}\big)^\dagger$.}
\label{fig:deconstruction1}
\end{subfigure}%
\begin{subfigure}{0.25\textwidth}
\begin{overpic}[width=\textwidth]{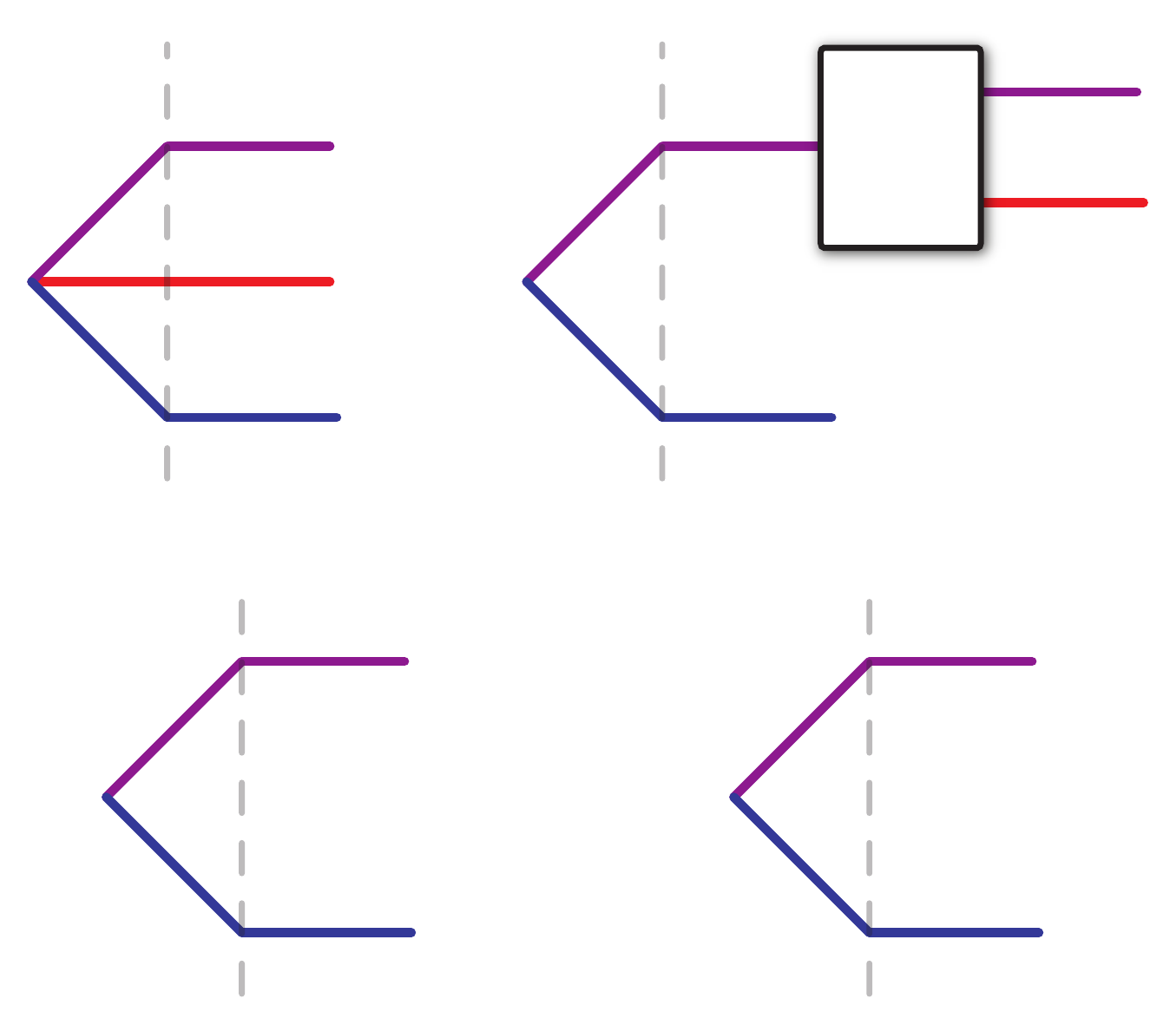}
  \put(16,76){\footnotesize$E$}
  \put(16,64.5){\footnotesize$AA'$}
  \put(16,53){\footnotesize$B$}
  \put(58,76){\footnotesize$E$}
  \put(58,53){\footnotesize$B$}
  \put(73.5,72.5){\footnotesize$\mathcal{R}$}
  \put(85,80.5){\footnotesize$E$}
  \put(85,71.3){\footnotesize$AA'$}
  \put(34,62){$\approx$}
  \put(12,42){\footnotesize$\omega$}
  \put(54,42){\footnotesize$\omega$}
  \put(46,17){$\approx$}
  \put(19,-1.5){\footnotesize$\rho$}
  \put(71.5,-1.5){\footnotesize$\omega$}
  \put(22,9){\footnotesize$B$}
  \put(22,32){\footnotesize$E$}
  \put(75,9){\footnotesize$B$}
  \put(75,32){\footnotesize$E$}
\end{overpic}
\caption{Local recoverability and negligible disturbance.}
\label{fig:deconstruction2}
\end{subfigure}
\caption{Depiction of (a) a state deconstruction protocol $\Lambda_{AA'E}$ with ancilla $\theta_{A'}$ for $\rho_{ABE}$ along with (b) the conditions of local recoverability $F(AA';B|E)_\omega\geq1-\varepsilon$ and negligible disturbance $F\left(\omega_{BE},\rho_{BE}\right)\geq1-\varepsilon$.}
\end{figure}

We note from Eq.~\eqref{eq:def_cqmi} that it is easy to see that the Groisman, Popescu, and Winter result can be invoked to say that $I(A;B|E)_{\rho}$ quantifies the additional cost to erase correlations between $A$ and $BE$ rather than just between $A$ and $E$. What has been missing so far, however, is a direct operational interpretation of the CQMI as a correlation measure in terms of quantum Markovianity. We now present exactly such an interpretation by extending the model of Groisman, Popescu, and Winter to incorporate a conditioning system $E$. We start with a tripartite quantum state $\rho_{ABE}$ and suppose that Alice holds $AE$ and Bob $B$. The task we want to accomplish is more delicate than just the total destruction of correlations between Alice and Bob. Namely, we are interested in the minimum rate of noise that Alice needs to apply to her systems such that
(i) the resulting system $A$ is \textit{locally recoverable} from the $E$ system alone, and
(ii) the correlations between $E$ and $B$ are only \textit{negligibly disturbed}.
We call the task a {\it state deconstruction protocol}, whose aim is to deconstruct (literally, to break into constituent components) the correlations in $\rho_{ABE}$. More precisely, a deconstruction protocol for $\rho_{ABE}$ is given by an already deconstructed, decoupled ancilla state $\theta_{A'}$, and a unitary randomizing channel
\begin{align}\label{eq:unital_randomizing}
\Lambda_{AA'E}(\cdot):=\frac{1}{M}\sum_{i=1}^MU^i_{AA'E}\big(\cdot\big)\big(U^i_{AA'E}\big)^\dagger \, ,
\end{align}
such that for the resulting state
\begin{align}\label{eq:omega_state}
\omega_{AA'BE}:=\Lambda_{AA'E}(\rho_{ABE}\otimes\theta_{A'})\,,
\end{align}
the above conditions (i) \& (ii) are fulfilled
\begin{align}\label{eq:conditions}
F(AA';B|E)_\omega\geq1-\varepsilon\quad\mathrm{\&}\quad F\left(\omega_{BE},\rho_{BE}\right)\geq1-\varepsilon\,.
\end{align}
A graphical depiction is presented in Figures~\ref{fig:deconstruction1} and \ref{fig:deconstruction2}.
The use of the ancilla system $A'$ is again catalytic in the sense that it is part of the output register and thus has to stay deconstructed with respect to $BE$ (at least approximately). We call the minimal rate of unitaries needed in the limit of many copies $\rho_{ABE}^{\otimes n}$ and vanishing error $\varepsilon\to0$ the {\it deconstruction cost of} $\rho_{ABE}$, denoted by $\mathcal{D}(A;B|E)_\rho$.


\paragraph{Conditional erasure of quantum correlations.} Alternatively, we can replace the local recoverability condition in~\eqref{eq:conditions} with the stronger condition
\begin{align}\label{eq:conditions_stronger}
F(\omega_{AA'BE},\pi_{AA'}\otimes\omega_{BE})\geq1-\varepsilon\,,
\end{align}
where $\pi_{AA'}$ denotes a maximally mixed state on a subspace of $AA'$. By choosing the local recovery channel as $\mathcal{R}_{E\rightarrow AA'E}(\cdot)=(\cdot)\otimes\omega_{AA'}$ we see that this new condition~\eqref{eq:conditions_stronger} surely implies the local recoverability condition in~\eqref{eq:conditions}. The {\it conditional erasure cost of} $\rho_{ABE}$, denoted by $\mathcal{C}(A;B|E)_\rho$, is then defined as the corresponding minimal rate of unitaries needed in the limit of many copies $\rho_{ABE}^{\otimes n}$ and vanishing error $\varepsilon\to0$. Thus, we have by definition $\mathcal{C}(A;B|E)_\rho\geq\mathcal{D}(A;B|E)_\rho$.


\paragraph{Conditional decoupling.}
Our models for deconstruction and conditional erasure extend the decoupling approach to quantum information theory~\cite{qip2002schu,horodecki05,qcap2008first,dupuis10,dupuis09} to a conditional version. While first conceived in the context of quantum channel \cite{qip2002schu} and source coding \cite{horodecki05}, the decoupling technique has numerous applications in areas as different as cryptography~\cite{Berta14}, quantum thermodynamics~\cite{RARDV10,Aberg13}, black hole radiation~\cite{HayPre07,braunstein-pati,braunstein-zyczkowski}, or many body quantum physics~\cite{Brandao12}.
Our models for deconstruction and conditional erasure extend this paradigm in the following sense. In conditional erasure, Alice does not want to erase all her correlations with Bob's system $B$ but only decouple her system $A$ from $B$ conditioned on the information she holds in system $E$, thereby not disturbing the correlations between $E$ and $B$. This negligible disturbance condition is critical: Alice and Bob might want to use their systems $E$ and $B$, respectively, for some later quantum information processing task, so that keeping the correlations intact is essential for the systems to be useful later on. The condition also highlights an essential difference between a semi-classical and fully quantum state deconstruction protocol: in the case that the system $E$ is classical, the negligible disturbance condition is not necessary because one could always observe the value without causing any disturbance to it. However, in the quantum case, the uncertainty principle forbids us from taking a similar action, so that it is necessary for a fully quantum state deconstruction protocol to proceed with a greater sleight of hand.


\paragraph{Main result.} It is the goal of this letter to show that both the deconstruction cost as well as the conditional erasure cost are given by the CQMI.

\begin{theorem}\label{thm:CQMI-upper-decons}
For any tripartite quantum state $\rho_{ABE}$:
\begin{align*}
\mathcal{D}(A;B|E)_\rho=I(A;B|E)_\rho=\mathcal{C}(A;B|E)_\rho\,.
\end{align*}
\end{theorem}

Thus, our result assigns a new physical meaning to the CQMI, in terms of an erasure or thermodynamical task that generalizes Landauer's original scenario as well as the erasure of correlations scenario of Groisman, Popescu, and Winter. The CQMI has many properties that are useful for a conditional measure of correlations. Amongst them are the duality property $I(A;B|R)_{\rho}=I(A;B|E)_{\rho}$ for a four party pure state $\rho_{ABER}$ and the chain rule
\begin{align}
I(A_{1}\cdots A_{n};B|E)_\rho=\sum_{i=1}^{n}I(A_{i};B|EA_{1}^{i-1})_\rho\label{eq:chainrule}
\end{align}
for $A_{1}^{i-1}:=A_{1}\cdots A_{i-1}$. The latter means that we can think of the correlations between $A_{1}\cdots A_{n}$ and $B$, as observed by $E$, being built up one system at a time.

We would like to emphasize again that deconstruction and conditional erasure protocols are more delicate than standard decoupling, the latter sometimes described as having the relatively indiscriminate goal of destruction~\cite{ADHW06FQSW}. That is, a straightforward application of the decoupling method is too blunt of a tool to apply in a state deconstruction protocol. Applying it naively would result in the annihilation of correlations such that if correlations between systems $B$ and $E$ were present beforehand, they would be destroyed.


\paragraph{Previous Work.} Our results are to be contrasted with the previous works of del Rio {\it et al.}~\cite{RARDV10} and Wakakuwa {\it et al.}~\cite{WSM15}. In~\cite{RARDV10} the authors give a conditional version of Landauer's erasure principle by showing that the work cost of {\it resetting} the $A$-part of a tripartite pure state $\rho_{ABR}$ to $\psi_A\otimes\rho_{BR}$ with $\psi_A$ pure, is given by the conditional entropy $H(A|B)_\rho$. There are various differences with our setting, but most importantly, we do not demand for the final state to be pure on $A$, but only that it is deconstructed as in~\eqref{eq:conditions} or decoupled and maximally mixed as in~\eqref{eq:conditions_stronger}. In~\cite{WSM15} the authors give an extension of the Groisman, Popescu, and Winter model~\eqref{eq:groisman1}--\eqref{eq:groisman2} to include a third system $E$. Their model, called {\it Markovianization cost}, is conceptually different from our models~\eqref{eq:unital_randomizing}--\eqref{eq:conditions_stronger} in various aspects: (i) their unitaries only act on $A$ and not on $AE$ (and hence there is no negligible disturbance condition on $BE$) (ii) the resulting state is asked to be close to an exact quantum Markov state~\cite{AF83} (however, see also~\cite{WSM15a}) (iii) there is no catalytic ancilla register. Whereas the converse from Proposition~\ref{prop:CQMI-lower-bnd} holds for their model as well~\cite{WSM15a}, the CQMI cannot be achieved: the different condition (i) accounts for a strictly larger optimal rate function based on the Koashi-Imoto decomposition~\cite{koashi02} (at least for pure states). This proves that the CQMI cannot be achieved without having access to the $E$ system (which is actually even true in the classical case~\cite{BBMW_full16}). Wakakuwa {\it et al.}'s result is motivated from questions in distributed computation~\cite{WSM15_2} but has the disadvantage that the Koashi-Imoto decomposition is not continuous in the state. We consider our models to be the most natural and refer to our companion paper~\cite{BBMW_full16} for an extended discussion.


\paragraph{Converse.} We only need to prove that the deconstruction cost of tripartite states is lower bounded by its CQMI since we have $\mathcal{C}(A;B|E)_\rho\geq\mathcal{D}(A;B|E)_\rho$. For that we make use of standard entropy inequalities and some properties of the FoR that are similar to the CQMI. In particular, the FoR is self-dual~\cite[Prop.~4]{SW14},
\begin{align}\label{eq:duality-FoR}
F(A;B|E)_{\rho}=F(A;B|R)_{\rho}\quad\text{for $\rho_{ABER}$ pure,}
\end{align}
and multiplicative on tensor-product states~\cite[Prop.~2]{BT15}. 

\begin{prop}\label{prop:CQMI-lower-bnd}
For any tripartite quantum state $\rho_{ABE}$:
\begin{align*}
\mathcal{D}(A;B|E)_\rho\geq I(A;B|E)_{\rho}\,.
\end{align*}
\end{prop}

\begin{proof}
Given an ancilla state $\theta_{A'}$ and a set of unitaries $\left\{U^i_{AA'E}\right\}_{i=1}^M$ leading to $\omega_{AA'BE}$ as in~\eqref{eq:omega_state}, we define an extended ancilla state $\theta_{A'A'_1A'_2}:=\theta_{A'}\otimes\tau_{A_1'A_2'}$ with each $\tau_{A_i'}$ maximally mixed of dimension $\sqrt{M}$~\footnote{Non-integer $\sqrt{M}$ can easily be taken care of as in~\cite{BBMW_full16}.}, and apply the unitaries $U^i_{AA'E}$ controlled on an orthonormal basis of maximally entangled states of $A_1'A_2'$. When tracing over $A_2'$, the resulting state is given by $\omega_{AA'BE}\otimes\tau_{A_1'}$ with $\omega_{AA'BE}$ from~\eqref{eq:omega_state}. Now, by the multiplicativity of the FoR we have $F(AA'A_1';B|E)_{\omega\otimes\tau}=F(AA';B|E)_\omega$, and hence we find that any lower bound on the size of the system $A_2'$ that has to be traced out in order to fulfill the conditions~\eqref{eq:conditions} for $\omega_{AA'BE}$, automatically gives a lower bound on the number $M$ of unitaries needed. To find a lower bound on $\left|A_2'\right|=\sqrt{M}$, we start with
\begin{align*}
nI(A;B|E)_{\rho}&
=I(A^{n}A'A'_1A'_2E^{n};B^{n})_{\rho^{\otimes n}\otimes\theta}-I(B^{n};E^{n})_{\rho^{\otimes n}}
\end{align*}
which follows because the CQMI is additive with respect to tensor-product states, invariant with respect to tensoring in a product state, and because of the CQMI chain rule~\eqref{eq:chainrule}. Now, we employ that the QMI is invariant with respect to local unitaries and that the QMI is continuous in the sense that
\begin{align*}
-I(B^{n};E^{n})_{\rho^{\otimes n}}\lesssim-I(B^{n};E^{n})_{\omega} \, ,
\end{align*}
with $\lesssim$ denoting an inequality that holds up to terms having order $n\sqrt{\varepsilon}$. From a dimension upper bound on the QMI (see, e.g., \cite{W16}), we then get
\begin{align*}
nI(A;B|E)_\rho\lesssim I(A^nA'A_1'E^{n};B^{n})_{\omega\otimes\tau}+2\log\left\vert A_2'\right\vert\,.
\end{align*}
Again using the additivity of the CQMI with respect to tensor-product states and the CQMI chain rule~\eqref{eq:chainrule}, we find that $I(A^nA'A_1'E^{n};B^{n})_{\omega\otimes\tau}=I(A^nA';B^{n}|E^{n})_{\omega}$. The claim follows by the converse of~\eqref{eq:fawzi_renner}, using, e.g., \cite[Prop.~35]{BSW14}, $F(A^nA';B^{n}|E^{n})_\omega\to1$ implies $I(A^nA';B^{n}|E^{n})_{\omega}\to0$) and by taking the limits $n\rightarrow\infty$ and $\varepsilon\to0$.
\end{proof}




\paragraph{Achievability.} We only need to prove that the conditional erasure cost of tripartite states is upper bounded by its CQMI since we have $\mathcal{D}(A;B|E)_\rho\leq\mathcal{C}(A;B|E)_\rho$.

\begin{prop}\label{prop:CQMI-upper-decons}
For any tripartite quantum state $\rho_{ABE}$:
\begin{align*}
\mathcal{C}(A;B|E)_\rho\leq I(A;B|E)_{\rho}\,.
\end{align*}
\end{prop}

We will make crucial use of a previously established operational interpretation of the CQMI in terms of {\it quantum state redistribution (QSR)}~\cite{DY08}. A QSR protocol begins with a sender, a receiver, and a reference party sharing many independent copies of a four party pure state $\rho_{ABER}$. The sender has $AE$, the receiver $R$, and the reference party $B$. The goal is to use noiseless quantum communication and entanglement assistance to redistribute the systems such that the sender ends up with $E$, the receiver with $AR$, and the reference keeps $B$. We will need the following key lemma from the follow-up work~\cite{PhysRevA.78.030302}, which shows that QSR is asymptotically achievable for a quantum communication rate of $\frac{1}{2}I(A;B|E)_\rho$, using entanglement assistance and a unitary encoder and decoder.

\begin{lemma}\cite[Thm.~3]{PhysRevA.78.030302}\label{lem:state_redistribution}
For every four party pure state $\rho_{ABER}$ there exist unitary operations $\mathrm{Enc}:A^nA'E^n\to A_0\bar{A}_0E^n$ and $\mathrm{Dec}:\bar{A}_0R^nR'\to A^nR_0R^n$ such that for $n\to\infty$ and maximally entangled states $\Phi_{A'R'}$ and $\Phi_{A_0R_0}$ of appropriate dimension,
\begin{align*}
F\left(\mathrm{Dec}\circ\mathrm{Enc}\left(\rho_{ABER}^{\otimes n}\otimes\Phi_{A'R'}\right),\rho_{ABER}^{\otimes n}\otimes\Phi_{A_0R_0}\right)\to1\,,
\end{align*}
with quantum communication $\frac{1}{n}\log|\bar{A}_0|\to\frac{1}{2}I(A;B|E)_\rho$.
\end{lemma}

We can now prove Proposition~\ref{prop:CQMI-upper-decons} by using the QSR encoder to construct the unitary randomizing channel~\eqref{eq:unital_randomizing}.


\begin{proof}[Proof of Proposition~\ref{prop:CQMI-upper-decons}.]
Let $\rho_{ABER}$ be a purification of $\rho_{ABE}$. We will show that there exists an ancilla register $\theta_{A'}$ with purification $\theta_{A'R'}$ and a unitary operation $\mathcal{V}_{A^nA'E^n\to A_0\bar{A}_0E^n}$ with $A^nA'\cong A_0\bar{A}_0$ such that for the resulting state
\begin{align}\label{eq:omega_new}
\omega_{A_0\bar{A}_0B^nE^nR^nR'}:=\mathcal{V}_{A^nA'E^n\to A_0\bar{A}_0E^n}\left(\rho_{ABER}^{\otimes n}\otimes\theta_{A'R'}\right)
\end{align}
we have in the limit $n\to\infty$,
\begin{align}\label{eq:cond}
F(\omega_{A_0B^nE^n},\pi_{A_0}\otimes\omega_{B^nE^n})\to1\;\mathrm{and}\;F\left(\omega_{B^nE^n},\rho_{BE}^{\otimes n}\right)\to1\,,
\end{align}
for the choice $\frac{1}{n}\log\left|\bar{A}_0\right|\to\frac{1}{2}I(A;B|E)_\rho$. From this we can pick the unitaries
\begin{align*}
U^i_{A^nA'E^n\to A_0\bar{A}_0E^n}:=W^i_{A_0\bar{A}_0E^n}V_{A^nA'E^n\to A_0\bar{A}_0E^n}\,,
\end{align*}
with $\{W^i_{A_0\bar{A}_0E^n}\}_{i=1}^{|\bar{A}_0|^2}$ a set of Heisenberg-Weyl unitaries that realize the partial trace over $\bar{A}_0$, and $V_{A^nA'E^n\to A_0\bar{A}_0E^n}$ implementing $\mathcal{V}_{A^nA'E^n\to A_0\bar{A}_0E^n}$. The set of unitaries
\begin{align*}
\text{$\{U^i_{A^nA'E^n\to A_0\bar{A}_0E^n}\}_{i=1}^M$ with $M=\left|\bar{A}_0\right|^2$}
\end{align*}
then defines a unitary randomizing channel $\Lambda_{A^nA'E^n\to A_0\bar{A}_0E^n}$ as in~\eqref{eq:unital_randomizing}, with the property
\begin{align*}
\Lambda_{A^nA'E^n\to A_0\bar{A}_0E^n}(\rho_{ABE}^{\otimes n}\otimes\theta_{A'})=\omega_{A_0B^nE^n}\otimes\tau_{\bar{A}_0}\,,
\end{align*}
and $\omega_{A_0B^nE^n}$ from~\eqref{eq:omega_new}. With~\eqref{eq:cond}, this implies the claim. Now, for $\mathcal{V}_{A^nA'E^n\to A_0\bar{A}_0E^n}$ we pick the QSR encoder for $\rho_{ABER}$ from Lemma~\ref{lem:state_redistribution},
\begin{align*}
\mathcal{V}_{A^nA'E^n\to A_0\bar{A}_0E^n}:=\mathrm{Enc}_{A^nA'E^n\to A_0\bar{A}_0E^n}\,,
\end{align*}
and furthermore we set $\theta_{A'R'}:=\Phi_{A'R'}$ maximally entangled. By Lemma~\ref{lem:state_redistribution} and the monotonicity of the fidelity under quantum operations we have $F\left(\omega_{B^nE^n},\rho_{BE}^{\otimes n}\right)\to1$. By the same monotonicity and the triangle inequality for any fidelity based metric, Lemma~\ref{lem:state_redistribution} implies $F(\omega_{A_0B^nE^n},\pi_{A_0}\otimes\omega_{B^nE^n})\to1$.
\end{proof}


\paragraph{Discussion.} The converse bound in Proposition~\ref{prop:CQMI-lower-bnd} together with the achievability bound in Proposition~\ref{prop:CQMI-upper-decons} provide a proof of our main result (Theorem~\ref{thm:CQMI-upper-decons}). This establishes the CQMI as an operational measure for the correlations between $A$ and $B$ from the perspective of $E$. Our result can alternatively be read as a conditional decoupling theorem and hence provides a conceptually new extension of the decoupling approach to quantum information theory. The power of decoupling lies in a fundamental monogamy of entanglement type duality that allows to retrieve quantum information from a purifying reference system if and only if it is decoupled~\cite{qip2002schu,horodecki05,qcap2008first,dupuis10,dupuis09}. In that sense, just as Groisman {\it et al.}'s destruction of bipartite correlations is dual to coherent quantum state merging~\cite{ADHW06FQSW,bertachristandl11,MBDRC16}, in our case we can make use of QSR, and in our companion paper~\cite{BBMW_full16}, we even show that the task of conditional erasure is equivalent to QSR. We emphasize that our negligible disturbance condition (ii) is exactly crucial for this duality to work in the tripartite setting.

More generally, the decoupling technique has numerous applications in areas as different as cryptography~\cite{Berta14}, quantum thermodynamics~\cite{RARDV10,Aberg13}, black hole radiation~\cite{HayPre07,braunstein-pati,braunstein-zyczkowski}, or many body quantum physics~\cite{Brandao12}. Hence, we expect our setting of conditional decoupling to have many more applications. In particular, since the CQMI serves as a measure for topological order~\cite{PK06,LW06,kim13}, it would be interesting to further explore this connection in terms of our findings. Another interesting avenue to explore on the information theory side is the connection of our conditional decoupling models to channel resolvability and wiretap channels (see, e.g., \cite{Hayashi2015} and \cite[Sect.~9.4 \& 9.5]{Hayashi2017}). Finally, the CQMI is also the basis of the correlation measures squashed entanglement~\cite{CW04} and quantum discord~\cite{zurek01}, and hence our result has immediate consequences for the study of these quantities. We discuss this in our companion paper~\cite{BBMW_full16}.


\paragraph{Conclusion.} We presented new operational interpretations of the CQMI as the deconstruction and conditional erasure cost of tripartite quantum states. Concerning open questions we would like to understand if the use of the catalytic ancillary register $A'$ is strictly necessary for achieving the CQMI. In our companion paper~\cite{BBMW_full16}, we show that for conditional erasure, our achievability result with a maximally mixed register $A'$ of rate
\begin{align*}
\frac{1}{n}\log|A'|\to\max\left\{\frac{1}{2}I(A:E)_\rho-\frac{1}{2}I(A:R)_\rho,0\right\}&\\
\text{for $\rho_{ABER}$ pure}&\,,
\end{align*}
is also optimal. However, for achieving the CQMI in state deconstruction only, the ancilla register might not be needed at all. We note that in the special case of Groisman {\it et al.}'s model~\eqref{eq:groisman1}--\eqref{eq:groisman2}, the ancilla register $A'$ is not needed in the asymptotic limit, but it seems to be useful for deriving tight one-shot bounds~\cite{MBDRC16}.


\begin{acknowledgments}
We are indebted to G.~Gour, M.~Hastings, M.~Piani, S.~Das, M.~Murao, K.~Seshadreesan, E.~Wakakuwa, and A.~Winter for valuable discussions. We acknowledge the catalyzing role of the open problems session at Beyond IID\ 2016, which ultimately led to the results presented here. CM acknowledges financial support from the European Research Council (ERC Grant Agreement no 337603), the Danish Council for Independent Research (Sapere Aude) and VILLUM FONDEN via the QMATH Centre of Excellence (Grant No. 10059). MMW\ acknowledges support from the NSF\ under Award no.~1714215.
\end{acknowledgments}


\bibliography{ref}

\begin{thebibliography}{55}%
\makeatletter
\providecommand \@ifxundefined [1]{%
 \@ifx{#1\undefined}
}%
\providecommand \@ifnum [1]{%
 \ifnum #1\expandafter \@firstoftwo
 \else \expandafter \@secondoftwo
 \fi
}%
\providecommand \@ifx [1]{%
 \ifx #1\expandafter \@firstoftwo
 \else \expandafter \@secondoftwo
 \fi
}%
\providecommand \natexlab [1]{#1}%
\providecommand \enquote  [1]{``#1''}%
\providecommand \bibnamefont  [1]{#1}%
\providecommand \bibfnamefont [1]{#1}%
\providecommand \citenamefont [1]{#1}%
\providecommand \href@noop [0]{\@secondoftwo}%
\providecommand \href [0]{\begingroup \@sanitize@url \@href}%
\providecommand \@href[1]{\@@startlink{#1}\@@href}%
\providecommand \@@href[1]{\endgroup#1\@@endlink}%
\providecommand \@sanitize@url [0]{\catcode `\\12\catcode `\$12\catcode
  `\&12\catcode `\#12\catcode `\^12\catcode `\_12\catcode `\%12\relax}%
\providecommand \@@startlink[1]{}%
\providecommand \@@endlink[0]{}%
\providecommand \url  [0]{\begingroup\@sanitize@url \@url }%
\providecommand \@url [1]{\endgroup\@href {#1}{\urlprefix }}%
\providecommand \urlprefix  [0]{URL }%
\providecommand \Eprint [0]{\href }%
\providecommand \doibase [0]{http://dx.doi.org/}%
\providecommand \selectlanguage [0]{\@gobble}%
\providecommand \bibinfo  [0]{\@secondoftwo}%
\providecommand \bibfield  [0]{\@secondoftwo}%
\providecommand \translation [1]{[#1]}%
\providecommand \BibitemOpen [0]{}%
\providecommand \bibitemStop [0]{}%
\providecommand \bibitemNoStop [0]{.\EOS\space}%
\providecommand \EOS [0]{\spacefactor3000\relax}%
\providecommand \BibitemShut  [1]{\csname bibitem#1\endcsname}%
\let\auto@bib@innerbib\@empty
\bibitem [{\citenamefont {Landauer}(1961)}]{L61}%
  \BibitemOpen
  \bibfield  {author} {\bibinfo {author} {\bibfnamefont {R.}~\bibnamefont
  {Landauer}},\ }\href {\doibase 10.1147/rd.53.0183} {\bibfield  {journal}
  {\bibinfo  {journal} {IBM Journal of Research and Development}\ }\textbf
  {\bibinfo {volume} {5}},\ \bibinfo {pages} {183} (\bibinfo {year}
  {1961})}\BibitemShut {NoStop}%
\bibitem [{\citenamefont {Groisman}\ \emph {et~al.}(2005)\citenamefont
  {Groisman}, \citenamefont {Popescu},\ and\ \citenamefont {Winter}}]{GPW05}%
  \BibitemOpen
  \bibfield  {author} {\bibinfo {author} {\bibfnamefont {B.}~\bibnamefont
  {Groisman}}, \bibinfo {author} {\bibfnamefont {S.}~\bibnamefont {Popescu}}, \
  and\ \bibinfo {author} {\bibfnamefont {A.}~\bibnamefont {Winter}},\ }\href
  {\doibase 10.1103/PhysRevA.72.032317} {\bibfield  {journal} {\bibinfo
  {journal} {Physical Review A}\ }\textbf {\bibinfo {volume} {72}},\ \bibinfo
  {pages} {032317} (\bibinfo {year} {2005})}\BibitemShut {NoStop}%
\bibitem [{\citenamefont {Faist}\ \emph {et~al.}(2015)\citenamefont {Faist},
  \citenamefont {Dupuis}, \citenamefont {Oppenheim},\ and\ \citenamefont
  {Renner}}]{faist2015minimal}%
  \BibitemOpen
  \bibfield  {author} {\bibinfo {author} {\bibfnamefont {P.}~\bibnamefont
  {Faist}}, \bibinfo {author} {\bibfnamefont {F.}~\bibnamefont {Dupuis}},
  \bibinfo {author} {\bibfnamefont {J.}~\bibnamefont {Oppenheim}}, \ and\
  \bibinfo {author} {\bibfnamefont {R.}~\bibnamefont {Renner}},\ }\href@noop {}
  {\bibfield  {journal} {\bibinfo  {journal} {Nature Communications}\ }\textbf
  {\bibinfo {volume} {6}} (\bibinfo {year} {2015})}\BibitemShut {NoStop}%
\bibitem [{\citenamefont {Maruyama}\ \emph {et~al.}(2009)\citenamefont
  {Maruyama}, \citenamefont {Nori},\ and\ \citenamefont
  {Vedral}}]{Maruyama2009}%
  \BibitemOpen
  \bibfield  {author} {\bibinfo {author} {\bibfnamefont {K.}~\bibnamefont
  {Maruyama}}, \bibinfo {author} {\bibfnamefont {F.}~\bibnamefont {Nori}}, \
  and\ \bibinfo {author} {\bibfnamefont {V.}~\bibnamefont {Vedral}},\
  }\href@noop {} {\bibfield  {journal} {\bibinfo  {journal} {Reviews of Modern
  Physics}\ }\textbf {\bibinfo {volume} {81}},\ \bibinfo {pages} {1} (\bibinfo
  {year} {2009})}\BibitemShut {NoStop}%
\bibitem [{\citenamefont {Plenio}\ and\ \citenamefont
  {Vitelli}(2001)}]{plenio2001physics}%
  \BibitemOpen
  \bibfield  {author} {\bibinfo {author} {\bibfnamefont {M.~B.}\ \bibnamefont
  {Plenio}}\ and\ \bibinfo {author} {\bibfnamefont {V.}~\bibnamefont
  {Vitelli}},\ }\href@noop {} {\bibfield  {journal} {\bibinfo  {journal}
  {Contemporary Physics}\ }\textbf {\bibinfo {volume} {42}},\ \bibinfo {pages}
  {25} (\bibinfo {year} {2001})}\BibitemShut {NoStop}%
\bibitem [{\citenamefont {B{\'e}rut}\ \emph {et~al.}(2012)\citenamefont
  {B{\'e}rut}, \citenamefont {Arakelyan}, \citenamefont {Petrosyan},
  \citenamefont {Ciliberto}, \citenamefont {Dillenschneider},\ and\
  \citenamefont {Lutz}}]{berut2012experimental}%
  \BibitemOpen
  \bibfield  {author} {\bibinfo {author} {\bibfnamefont {A.}~\bibnamefont
  {B{\'e}rut}}, \bibinfo {author} {\bibfnamefont {A.}~\bibnamefont
  {Arakelyan}}, \bibinfo {author} {\bibfnamefont {A.}~\bibnamefont
  {Petrosyan}}, \bibinfo {author} {\bibfnamefont {S.}~\bibnamefont
  {Ciliberto}}, \bibinfo {author} {\bibfnamefont {R.}~\bibnamefont
  {Dillenschneider}}, \ and\ \bibinfo {author} {\bibfnamefont {E.}~\bibnamefont
  {Lutz}},\ }\href@noop {} {\bibfield  {journal} {\bibinfo  {journal} {Nature}\
  }\textbf {\bibinfo {volume} {483}},\ \bibinfo {pages} {187} (\bibinfo {year}
  {2012})}\BibitemShut {NoStop}%
\bibitem [{Note1()}]{Note1}%
  \BibitemOpen
  \bibinfo {note} {Groisman, Popescu, and Winter discuss various models of how
  to inject noise into the system; however, ultimately all of them become
  equivalent.}\BibitemShut {Stop}%
\bibitem [{\citenamefont {Majenz}\ \emph {et~al.}(2017)\citenamefont {Majenz},
  \citenamefont {Berta}, \citenamefont {Dupuis}, \citenamefont {Renner},\ and\
  \citenamefont {Christandl}}]{MBDRC16}%
  \BibitemOpen
  \bibfield  {author} {\bibinfo {author} {\bibfnamefont {C.}~\bibnamefont
  {Majenz}}, \bibinfo {author} {\bibfnamefont {M.}~\bibnamefont {Berta}},
  \bibinfo {author} {\bibfnamefont {F.}~\bibnamefont {Dupuis}}, \bibinfo
  {author} {\bibfnamefont {R.}~\bibnamefont {Renner}}, \ and\ \bibinfo {author}
  {\bibfnamefont {M.}~\bibnamefont {Christandl}},\ }\href {\doibase
  10.1103/PhysRevLett.118.080503} {\bibfield  {journal} {\bibinfo  {journal}
  {Physical Review Letters}\ }\textbf {\bibinfo {volume} {118}},\ \bibinfo
  {pages} {080503} (\bibinfo {year} {2017})}\BibitemShut {NoStop}%
\bibitem [{\citenamefont {Hayden}\ and\ \citenamefont
  {Preskill}(2007)}]{HayPre07}%
  \BibitemOpen
  \bibfield  {author} {\bibinfo {author} {\bibfnamefont {P.}~\bibnamefont
  {Hayden}}\ and\ \bibinfo {author} {\bibfnamefont {J.}~\bibnamefont
  {Preskill}},\ }\href@noop {} {\bibfield  {journal} {\bibinfo  {journal}
  {Journal of High Energy Physics}\ }\textbf {\bibinfo {volume} {07}},\
  \bibinfo {pages} {120} (\bibinfo {year} {2007})}\BibitemShut {NoStop}%
\bibitem [{\citenamefont {Braunstein}\ and\ \citenamefont
  {Pati}(2007)}]{braunstein-pati}%
  \BibitemOpen
  \bibfield  {author} {\bibinfo {author} {\bibfnamefont {S.~L.}\ \bibnamefont
  {Braunstein}}\ and\ \bibinfo {author} {\bibfnamefont {A.~K.}\ \bibnamefont
  {Pati}},\ }\href@noop {} {\bibfield  {journal} {\bibinfo  {journal} {Physical
  Review Letters}\ }\textbf {\bibinfo {volume} {98}},\ \bibinfo {pages}
  {080502} (\bibinfo {year} {2007})}\BibitemShut {NoStop}%
\bibitem [{\citenamefont {Braunstein}\ \emph {et~al.}(2013)\citenamefont
  {Braunstein}, \citenamefont {Pirandola},\ and\ \citenamefont
  {Zyczkowski}}]{braunstein-zyczkowski}%
  \BibitemOpen
  \bibfield  {author} {\bibinfo {author} {\bibfnamefont {S.~L.}\ \bibnamefont
  {Braunstein}}, \bibinfo {author} {\bibfnamefont {S.}~\bibnamefont
  {Pirandola}}, \ and\ \bibinfo {author} {\bibfnamefont {K.}~\bibnamefont
  {Zyczkowski}},\ }\href@noop {} {\bibfield  {journal} {\bibinfo  {journal}
  {Physical Review Letters}\ }\textbf {\bibinfo {volume} {110}},\ \bibinfo
  {pages} {101301} (\bibinfo {year} {2013})}\BibitemShut {NoStop}%
\bibitem [{\citenamefont {Brandao}\ and\ \citenamefont
  {Horodecki}(2013)}]{Brandao12}%
  \BibitemOpen
  \bibfield  {author} {\bibinfo {author} {\bibfnamefont {F.~G. S.~L.}\
  \bibnamefont {Brandao}}\ and\ \bibinfo {author} {\bibfnamefont
  {M.}~\bibnamefont {Horodecki}},\ }\href@noop {} {\bibfield  {journal}
  {\bibinfo  {journal} {Nature Physics}\ }\textbf {\bibinfo {volume} {9}},\
  \bibinfo {pages} {721} (\bibinfo {year} {2013})}\BibitemShut {NoStop}%
\bibitem [{\citenamefont {Lieb}\ and\ \citenamefont
  {Ruskai}(1973)}]{PhysRevLett.30.434}%
  \BibitemOpen
  \bibfield  {author} {\bibinfo {author} {\bibfnamefont {E.~H.}\ \bibnamefont
  {Lieb}}\ and\ \bibinfo {author} {\bibfnamefont {M.~B.}\ \bibnamefont
  {Ruskai}},\ }\href {\doibase 10.1103/PhysRevLett.30.434} {\bibfield
  {journal} {\bibinfo  {journal} {Physical Review Letters}\ }\textbf {\bibinfo
  {volume} {30}},\ \bibinfo {pages} {434} (\bibinfo {year} {1973})}\BibitemShut
  {NoStop}%
\bibitem [{\citenamefont {Devetak}\ and\ \citenamefont {Yard}(2008)}]{DY08}%
  \BibitemOpen
  \bibfield  {author} {\bibinfo {author} {\bibfnamefont {I.}~\bibnamefont
  {Devetak}}\ and\ \bibinfo {author} {\bibfnamefont {J.}~\bibnamefont {Yard}},\
  }\href {\doibase 10.1103/PhysRevLett.100.230501} {\bibfield  {journal}
  {\bibinfo  {journal} {Physical Review Letters}\ }\textbf {\bibinfo {volume}
  {100}},\ \bibinfo {pages} {230501} (\bibinfo {year} {2008})}\BibitemShut
  {NoStop}%
\bibitem [{\citenamefont {Tomamichel}\ and\ \citenamefont
  {Hayashi}(2018)}]{Tomamichel2018}%
  \BibitemOpen
  \bibfield  {author} {\bibinfo {author} {\bibfnamefont {M.}~\bibnamefont
  {Tomamichel}}\ and\ \bibinfo {author} {\bibfnamefont {M.}~\bibnamefont
  {Hayashi}},\ }\href@noop {} {\bibfield  {journal} {\bibinfo  {journal} {IEEE
  Transactions on Information Theory}\ }\textbf {\bibinfo {volume} {64}},\
  \bibinfo {pages} {1064} (\bibinfo {year} {2018})}\BibitemShut {NoStop}%
\bibitem [{\citenamefont {Cooney}\ \emph {et~al.}(2016)\citenamefont {Cooney},
  \citenamefont {Hirche}, \citenamefont {Morgan}, \citenamefont {Olson},
  \citenamefont {Seshadreesan}, \citenamefont {Watrous},\ and\ \citenamefont
  {Wilde}}]{PhysRevA.94.022310}%
  \BibitemOpen
  \bibfield  {author} {\bibinfo {author} {\bibfnamefont {T.}~\bibnamefont
  {Cooney}}, \bibinfo {author} {\bibfnamefont {C.}~\bibnamefont {Hirche}},
  \bibinfo {author} {\bibfnamefont {C.}~\bibnamefont {Morgan}}, \bibinfo
  {author} {\bibfnamefont {J.~P.}\ \bibnamefont {Olson}}, \bibinfo {author}
  {\bibfnamefont {K.~P.}\ \bibnamefont {Seshadreesan}}, \bibinfo {author}
  {\bibfnamefont {J.}~\bibnamefont {Watrous}}, \ and\ \bibinfo {author}
  {\bibfnamefont {M.~M.}\ \bibnamefont {Wilde}},\ }\href {\doibase
  10.1103/PhysRevA.94.022310} {\bibfield  {journal} {\bibinfo  {journal}
  {Physical Review A}\ }\textbf {\bibinfo {volume} {94}},\ \bibinfo {pages}
  {022310} (\bibinfo {year} {2016})}\BibitemShut {NoStop}%
\bibitem [{\citenamefont {Berta}\ \emph {et~al.}(2017)\citenamefont {Berta},
  \citenamefont {Brandao},\ and\ \citenamefont {Hirche}}]{Berta2017}%
  \BibitemOpen
  \bibfield  {author} {\bibinfo {author} {\bibfnamefont {M.}~\bibnamefont
  {Berta}}, \bibinfo {author} {\bibfnamefont {F.~G.}\ \bibnamefont {Brandao}},
  \ and\ \bibinfo {author} {\bibfnamefont {C.}~\bibnamefont {Hirche}},\
  }\href@noop {} {\bibfield  {journal} {\bibinfo  {journal} {arXiv:1709.07268}\
  } (\bibinfo {year} {2017})}\BibitemShut {NoStop}%
\bibitem [{\citenamefont {Zeng}\ \emph {et~al.}(2015)\citenamefont {Zeng},
  \citenamefont {Chen}, \citenamefont {Zhou},\ and\ \citenamefont
  {Wen}}]{Zeng2015}%
  \BibitemOpen
  \bibfield  {author} {\bibinfo {author} {\bibfnamefont {B.}~\bibnamefont
  {Zeng}}, \bibinfo {author} {\bibfnamefont {X.}~\bibnamefont {Chen}}, \bibinfo
  {author} {\bibfnamefont {D.-L.}\ \bibnamefont {Zhou}}, \ and\ \bibinfo
  {author} {\bibfnamefont {X.-G.}\ \bibnamefont {Wen}},\ }\href@noop {}
  {\bibfield  {journal} {\bibinfo  {journal} {arXiv:1508.02595}\ } (\bibinfo
  {year} {2015})}\BibitemShut {NoStop}%
\bibitem [{\citenamefont {Kim}(2012)}]{Kim2012}%
  \BibitemOpen
  \bibfield  {author} {\bibinfo {author} {\bibfnamefont {I.~H.}\ \bibnamefont
  {Kim}},\ }\href@noop {} {\bibfield  {journal} {\bibinfo  {journal} {Physical
  Review B}\ }\textbf {\bibinfo {volume} {86}},\ \bibinfo {pages} {245116}
  (\bibinfo {year} {2012})}\BibitemShut {NoStop}%
\bibitem [{\citenamefont {Pastawski}\ \emph {et~al.}(2017)\citenamefont
  {Pastawski}, \citenamefont {Eisert},\ and\ \citenamefont
  {Wilming}}]{Pastawski2017}%
  \BibitemOpen
  \bibfield  {author} {\bibinfo {author} {\bibfnamefont {F.}~\bibnamefont
  {Pastawski}}, \bibinfo {author} {\bibfnamefont {J.}~\bibnamefont {Eisert}}, \
  and\ \bibinfo {author} {\bibfnamefont {H.}~\bibnamefont {Wilming}},\
  }\href@noop {} {\bibfield  {journal} {\bibinfo  {journal} {Physical Review
  Letters}\ }\textbf {\bibinfo {volume} {119}},\ \bibinfo {pages} {020501}
  (\bibinfo {year} {2017})}\BibitemShut {NoStop}%
\bibitem [{\citenamefont {Czech}\ \emph {et~al.}(2015)\citenamefont {Czech},
  \citenamefont {Lamprou}, \citenamefont {McCandlish},\ and\ \citenamefont
  {Sully}}]{Czech2015}%
  \BibitemOpen
  \bibfield  {author} {\bibinfo {author} {\bibfnamefont {B.}~\bibnamefont
  {Czech}}, \bibinfo {author} {\bibfnamefont {L.}~\bibnamefont {Lamprou}},
  \bibinfo {author} {\bibfnamefont {S.}~\bibnamefont {McCandlish}}, \ and\
  \bibinfo {author} {\bibfnamefont {J.}~\bibnamefont {Sully}},\ }\href@noop {}
  {\bibfield  {journal} {\bibinfo  {journal} {Journal of High Energy Physics}\
  }\textbf {\bibinfo {volume} {2015}},\ \bibinfo {pages} {175} (\bibinfo {year}
  {2015})}\BibitemShut {NoStop}%
\bibitem [{\citenamefont {Mahajan}\ \emph {et~al.}(2016)\citenamefont
  {Mahajan}, \citenamefont {Freeman}, \citenamefont {Mumford}, \citenamefont
  {Tubman},\ and\ \citenamefont {Swingle}}]{Mahajan2016}%
  \BibitemOpen
  \bibfield  {author} {\bibinfo {author} {\bibfnamefont {R.}~\bibnamefont
  {Mahajan}}, \bibinfo {author} {\bibfnamefont {C.~D.}\ \bibnamefont
  {Freeman}}, \bibinfo {author} {\bibfnamefont {S.}~\bibnamefont {Mumford}},
  \bibinfo {author} {\bibfnamefont {N.}~\bibnamefont {Tubman}}, \ and\ \bibinfo
  {author} {\bibfnamefont {B.}~\bibnamefont {Swingle}},\ }\href@noop {}
  {\bibfield  {journal} {\bibinfo  {journal} {arXiv:1608.05074}\ } (\bibinfo
  {year} {2016})}\BibitemShut {NoStop}%
\bibitem [{\citenamefont {Bettencourt}\ \emph {et~al.}(2008)\citenamefont
  {Bettencourt}, \citenamefont {Gintautas},\ and\ \citenamefont
  {Ham}}]{Bettencourt2008}%
  \BibitemOpen
  \bibfield  {author} {\bibinfo {author} {\bibfnamefont {L.~M.}\ \bibnamefont
  {Bettencourt}}, \bibinfo {author} {\bibfnamefont {V.}~\bibnamefont
  {Gintautas}}, \ and\ \bibinfo {author} {\bibfnamefont {M.~I.}\ \bibnamefont
  {Ham}},\ }\href@noop {} {\bibfield  {journal} {\bibinfo  {journal} {Physical
  Review Letters}\ }\textbf {\bibinfo {volume} {100}},\ \bibinfo {pages}
  {238701} (\bibinfo {year} {2008})}\BibitemShut {NoStop}%
\bibitem [{\citenamefont {Petz}(1986)}]{Pet86}%
  \BibitemOpen
  \bibfield  {author} {\bibinfo {author} {\bibfnamefont {D.}~\bibnamefont
  {Petz}},\ }\href {\doibase 10.1007/BF01212345} {\bibfield  {journal}
  {\bibinfo  {journal} {Communications in Mathematical Physics}\ }\textbf
  {\bibinfo {volume} {105}},\ \bibinfo {pages} {123} (\bibinfo {year}
  {1986})}\BibitemShut {NoStop}%
\bibitem [{\citenamefont {Seshadreesan}\ and\ \citenamefont
  {Wilde}(2015)}]{SW14}%
  \BibitemOpen
  \bibfield  {author} {\bibinfo {author} {\bibfnamefont {K.~P.}\ \bibnamefont
  {Seshadreesan}}\ and\ \bibinfo {author} {\bibfnamefont {M.~M.}\ \bibnamefont
  {Wilde}},\ }\href@noop {} {\bibfield  {journal} {\bibinfo  {journal}
  {Physical Review A}\ }\textbf {\bibinfo {volume} {92}},\ \bibinfo {pages}
  {042321} (\bibinfo {year} {2015})}\BibitemShut {NoStop}%
\bibitem [{\citenamefont {Fawzi}\ and\ \citenamefont {Renner}(2015)}]{FR14}%
  \BibitemOpen
  \bibfield  {author} {\bibinfo {author} {\bibfnamefont {O.}~\bibnamefont
  {Fawzi}}\ and\ \bibinfo {author} {\bibfnamefont {R.}~\bibnamefont {Renner}},\
  }\href@noop {} {\bibfield  {journal} {\bibinfo  {journal} {Communications in
  Mathematical Physics}\ }\textbf {\bibinfo {volume} {340}},\ \bibinfo {pages}
  {575} (\bibinfo {year} {2015})}\BibitemShut {NoStop}%
\bibitem [{\citenamefont {Levin}\ and\ \citenamefont {Wen}(2006)}]{LW06}%
  \BibitemOpen
  \bibfield  {author} {\bibinfo {author} {\bibfnamefont {M.}~\bibnamefont
  {Levin}}\ and\ \bibinfo {author} {\bibfnamefont {X.-G.}\ \bibnamefont
  {Wen}},\ }\href {\doibase 10.1103/PhysRevLett.96.110405} {\bibfield
  {journal} {\bibinfo  {journal} {Physical Review Letters}\ }\textbf {\bibinfo
  {volume} {96}},\ \bibinfo {pages} {110405} (\bibinfo {year}
  {2006})}\BibitemShut {NoStop}%
\bibitem [{\citenamefont {Kitaev}\ and\ \citenamefont {Preskill}(2006)}]{PK06}%
  \BibitemOpen
  \bibfield  {author} {\bibinfo {author} {\bibfnamefont {A.}~\bibnamefont
  {Kitaev}}\ and\ \bibinfo {author} {\bibfnamefont {J.}~\bibnamefont
  {Preskill}},\ }\href {\doibase 10.1103/PhysRevLett.96.110404} {\bibfield
  {journal} {\bibinfo  {journal} {Physical Review Letters}\ }\textbf {\bibinfo
  {volume} {96}},\ \bibinfo {pages} {110404} (\bibinfo {year}
  {2006})}\BibitemShut {NoStop}%
\bibitem [{\citenamefont {Lashkari}(2017)}]{L17}%
  \BibitemOpen
  \bibfield  {author} {\bibinfo {author} {\bibfnamefont {N.}~\bibnamefont
  {Lashkari}},\ }\href@noop {} {\  (\bibinfo {year} {2017})},\ \bibinfo {note}
  {arXiv:1704.05077}\BibitemShut {NoStop}%
\bibitem [{\citenamefont {Schumacher}\ and\ \citenamefont
  {Westmoreland}(2002)}]{qip2002schu}%
  \BibitemOpen
  \bibfield  {author} {\bibinfo {author} {\bibfnamefont {B.}~\bibnamefont
  {Schumacher}}\ and\ \bibinfo {author} {\bibfnamefont {M.~D.}\ \bibnamefont
  {Westmoreland}},\ }\href@noop {} {\bibfield  {journal} {\bibinfo  {journal}
  {Quantum Information Processing}\ }\textbf {\bibinfo {volume} {1}},\ \bibinfo
  {pages} {5} (\bibinfo {year} {2002})},\ \bibinfo {note}
  {arXiv:quant-ph/0112106}\BibitemShut {NoStop}%
\bibitem [{\citenamefont {Horodecki}\ \emph {et~al.}(2005)\citenamefont
  {Horodecki}, \citenamefont {Oppenheim},\ and\ \citenamefont
  {Winter}}]{horodecki05}%
  \BibitemOpen
  \bibfield  {author} {\bibinfo {author} {\bibfnamefont {M.}~\bibnamefont
  {Horodecki}}, \bibinfo {author} {\bibfnamefont {J.}~\bibnamefont
  {Oppenheim}}, \ and\ \bibinfo {author} {\bibfnamefont {A.}~\bibnamefont
  {Winter}},\ }\href {\doibase 10.1038/nature03909} {\bibfield  {journal}
  {\bibinfo  {journal} {Nature}\ }\textbf {\bibinfo {volume} {436}},\ \bibinfo
  {pages} {673} (\bibinfo {year} {2005})}\BibitemShut {NoStop}%
\bibitem [{\citenamefont {Hayden}\ \emph {et~al.}(2008)\citenamefont {Hayden},
  \citenamefont {Horodecki}, \citenamefont {Winter},\ and\ \citenamefont
  {Yard}}]{qcap2008first}%
  \BibitemOpen
  \bibfield  {author} {\bibinfo {author} {\bibfnamefont {P.}~\bibnamefont
  {Hayden}}, \bibinfo {author} {\bibfnamefont {M.}~\bibnamefont {Horodecki}},
  \bibinfo {author} {\bibfnamefont {A.}~\bibnamefont {Winter}}, \ and\ \bibinfo
  {author} {\bibfnamefont {J.}~\bibnamefont {Yard}},\ }\href@noop {} {\bibfield
   {journal} {\bibinfo  {journal} {Open Systems \& Information Dynamics}\
  }\textbf {\bibinfo {volume} {15}},\ \bibinfo {pages} {7} (\bibinfo {year}
  {2008})}\BibitemShut {NoStop}%
\bibitem [{\citenamefont {Dupuis}\ \emph {et~al.}(2014)\citenamefont {Dupuis},
  \citenamefont {Berta}, \citenamefont {Wullschleger},\ and\ \citenamefont
  {Renner}}]{dupuis10}%
  \BibitemOpen
  \bibfield  {author} {\bibinfo {author} {\bibfnamefont {F.}~\bibnamefont
  {Dupuis}}, \bibinfo {author} {\bibfnamefont {M.}~\bibnamefont {Berta}},
  \bibinfo {author} {\bibfnamefont {J.}~\bibnamefont {Wullschleger}}, \ and\
  \bibinfo {author} {\bibfnamefont {R.}~\bibnamefont {Renner}},\ }\href
  {\doibase 10.1007/s00220-014-1990-4} {\bibfield  {journal} {\bibinfo
  {journal} {Communications in Mathematical Physics}\ }\textbf {\bibinfo
  {volume} {328}},\ \bibinfo {pages} {251} (\bibinfo {year}
  {2014})}\BibitemShut {NoStop}%
\bibitem [{\citenamefont {Dupuis}(2009)}]{dupuis09}%
  \BibitemOpen
  \bibfield  {author} {\bibinfo {author} {\bibfnamefont {F.}~\bibnamefont
  {Dupuis}},\ }\emph {\bibinfo {title} {{The Decoupling Approach to Quantum
  Information Theory}}},\ \href@noop {} {Ph.D. thesis},\ \bibinfo  {school}
  {Universit{\'{e}} de Montr{\'{e}}al} (\bibinfo {year} {2009})\BibitemShut
  {NoStop}%
\bibitem [{\citenamefont {Berta}\ \emph {et~al.}(2014)\citenamefont {Berta},
  \citenamefont {Fawzi},\ and\ \citenamefont {Wehner}}]{Berta14}%
  \BibitemOpen
  \bibfield  {author} {\bibinfo {author} {\bibfnamefont {M.}~\bibnamefont
  {Berta}}, \bibinfo {author} {\bibfnamefont {O.}~\bibnamefont {Fawzi}}, \ and\
  \bibinfo {author} {\bibfnamefont {S.}~\bibnamefont {Wehner}},\ }\href@noop {}
  {\bibfield  {journal} {\bibinfo  {journal} {Information Theory, IEEE
  Transactions on}\ }\textbf {\bibinfo {volume} {60}},\ \bibinfo {pages} {1168}
  (\bibinfo {year} {2014})}\BibitemShut {NoStop}%
\bibitem [{\citenamefont {del Rio}\ \emph {et~al.}(2011)\citenamefont {del
  Rio}, \citenamefont {\AA{}berg}, \citenamefont {Renner}, \citenamefont
  {Dahlsten},\ and\ \citenamefont {Vedral}}]{RARDV10}%
  \BibitemOpen
  \bibfield  {author} {\bibinfo {author} {\bibfnamefont {L.}~\bibnamefont {del
  Rio}}, \bibinfo {author} {\bibfnamefont {J.}~\bibnamefont {\AA{}berg}},
  \bibinfo {author} {\bibfnamefont {R.}~\bibnamefont {Renner}}, \bibinfo
  {author} {\bibfnamefont {O.}~\bibnamefont {Dahlsten}}, \ and\ \bibinfo
  {author} {\bibfnamefont {V.}~\bibnamefont {Vedral}},\ }\href@noop {}
  {\bibfield  {journal} {\bibinfo  {journal} {Nature}\ }\textbf {\bibinfo
  {volume} {474}},\ \bibinfo {pages} {61} (\bibinfo {year} {2011})}\BibitemShut
  {NoStop}%
\bibitem [{\citenamefont {Aberg}(2013)}]{Aberg13}%
  \BibitemOpen
  \bibfield  {author} {\bibinfo {author} {\bibfnamefont {J.}~\bibnamefont
  {Aberg}},\ }\href@noop {} {\bibfield  {journal} {\bibinfo  {journal} {Nature
  Communications}\ }\textbf {\bibinfo {volume} {4}},\ \bibinfo {pages} {1925}
  (\bibinfo {year} {2013})}\BibitemShut {NoStop}%
\bibitem [{\citenamefont {Abeyesinghe}\ \emph {et~al.}(2009)\citenamefont
  {Abeyesinghe}, \citenamefont {Devetak}, \citenamefont {Hayden},\ and\
  \citenamefont {Winter}}]{ADHW06FQSW}%
  \BibitemOpen
  \bibfield  {author} {\bibinfo {author} {\bibfnamefont {A.}~\bibnamefont
  {Abeyesinghe}}, \bibinfo {author} {\bibfnamefont {I.}~\bibnamefont
  {Devetak}}, \bibinfo {author} {\bibfnamefont {P.}~\bibnamefont {Hayden}}, \
  and\ \bibinfo {author} {\bibfnamefont {A.}~\bibnamefont {Winter}},\
  }\href@noop {} {\bibfield  {journal} {\bibinfo  {journal} {Proceedings of the
  Royal Society A}\ }\textbf {\bibinfo {volume} {465}},\ \bibinfo {pages}
  {2537} (\bibinfo {year} {2009})}\BibitemShut {NoStop}%
\bibitem [{\citenamefont {Wakakuwa}\ \emph
  {et~al.}(2017{\natexlab{a}})\citenamefont {Wakakuwa}, \citenamefont {Soeda},\
  and\ \citenamefont {Murao}}]{WSM15}%
  \BibitemOpen
  \bibfield  {author} {\bibinfo {author} {\bibfnamefont {E.}~\bibnamefont
  {Wakakuwa}}, \bibinfo {author} {\bibfnamefont {A.}~\bibnamefont {Soeda}}, \
  and\ \bibinfo {author} {\bibfnamefont {M.}~\bibnamefont {Murao}},\
  }\href@noop {} {\bibfield  {journal} {\bibinfo  {journal} {IEEE Transactions
  on Information Theory}\ }\textbf {\bibinfo {volume} {63}},\ \bibinfo {pages}
  {1280} (\bibinfo {year} {2017}{\natexlab{a}})}\BibitemShut {NoStop}%
\bibitem [{\citenamefont {Accardi}\ and\ \citenamefont
  {Frigerio}(1983)}]{AF83}%
  \BibitemOpen
  \bibfield  {author} {\bibinfo {author} {\bibfnamefont {L.}~\bibnamefont
  {Accardi}}\ and\ \bibinfo {author} {\bibfnamefont {A.}~\bibnamefont
  {Frigerio}},\ }\href@noop {} {\bibfield  {journal} {\bibinfo  {journal}
  {Proceedings of the Royal Irish Academy. Section A: Mathematical and Physical
  Sciences}\ }\textbf {\bibinfo {volume} {83A}},\ \bibinfo {pages} {251}
  (\bibinfo {year} {1983})}\BibitemShut {NoStop}%
\bibitem [{\citenamefont {Wakakuwa}\ \emph
  {et~al.}(2017{\natexlab{b}})\citenamefont {Wakakuwa}, \citenamefont {Soeda},\
  and\ \citenamefont {Murao}}]{WSM15a}%
  \BibitemOpen
  \bibfield  {author} {\bibinfo {author} {\bibfnamefont {E.}~\bibnamefont
  {Wakakuwa}}, \bibinfo {author} {\bibfnamefont {A.}~\bibnamefont {Soeda}}, \
  and\ \bibinfo {author} {\bibfnamefont {M.}~\bibnamefont {Murao}},\
  }\href@noop {} {\bibfield  {journal} {\bibinfo  {journal} {IEEE Transactions
  on Information Theory}\ }\textbf {\bibinfo {volume} {63}},\ \bibinfo {pages}
  {5360} (\bibinfo {year} {2017}{\natexlab{b}})}\BibitemShut {NoStop}%
\bibitem [{\citenamefont {Koashi}\ and\ \citenamefont
  {Imoto}(2002)}]{koashi02}%
  \BibitemOpen
  \bibfield  {author} {\bibinfo {author} {\bibfnamefont {M.}~\bibnamefont
  {Koashi}}\ and\ \bibinfo {author} {\bibfnamefont {N.}~\bibnamefont {Imoto}},\
  }\href {\doibase 10.1103/PhysRevA.66.022318} {\bibfield  {journal} {\bibinfo
  {journal} {Physical Review A}\ }\textbf {\bibinfo {volume} {66}},\ \bibinfo
  {pages} {022318} (\bibinfo {year} {2002})}\BibitemShut {NoStop}%
\bibitem [{\citenamefont {Berta}\ \emph {et~al.}(2016)\citenamefont {Berta},
  \citenamefont {Majenz}, \citenamefont {Brandao},\ and\ \citenamefont
  {Wilde}}]{BBMW_full16}%
  \BibitemOpen
  \bibfield  {author} {\bibinfo {author} {\bibfnamefont {M.}~\bibnamefont
  {Berta}}, \bibinfo {author} {\bibfnamefont {C.}~\bibnamefont {Majenz}},
  \bibinfo {author} {\bibfnamefont {F.~G. S.~L.}\ \bibnamefont {Brandao}}, \
  and\ \bibinfo {author} {\bibfnamefont {M.~M.}\ \bibnamefont {Wilde}},\
  }\href@noop {} {\  (\bibinfo {year} {2016})},\ \bibinfo {note}
  {arXiv:1609.06994}\BibitemShut {NoStop}%
\bibitem [{\citenamefont {Wakakuwa}\ \emph
  {et~al.}(2017{\natexlab{c}})\citenamefont {Wakakuwa}, \citenamefont {Soeda},\
  and\ \citenamefont {Murao}}]{WSM15_2}%
  \BibitemOpen
  \bibfield  {author} {\bibinfo {author} {\bibfnamefont {E.}~\bibnamefont
  {Wakakuwa}}, \bibinfo {author} {\bibfnamefont {A.}~\bibnamefont {Soeda}}, \
  and\ \bibinfo {author} {\bibfnamefont {M.}~\bibnamefont {Murao}},\
  }\href@noop {} {\bibfield  {journal} {\bibinfo  {journal} {IEEE Transactions
  on Information Theory}\ }\textbf {\bibinfo {volume} {63}},\ \bibinfo {pages}
  {5372} (\bibinfo {year} {2017}{\natexlab{c}})}\BibitemShut {NoStop}%
\bibitem [{\citenamefont {Berta}\ and\ \citenamefont
  {Tomamichel}(2016)}]{BT15}%
  \BibitemOpen
  \bibfield  {author} {\bibinfo {author} {\bibfnamefont {M.}~\bibnamefont
  {Berta}}\ and\ \bibinfo {author} {\bibfnamefont {M.}~\bibnamefont
  {Tomamichel}},\ }\href {\doibase 10.1109/TIT.2016.2527683} {\bibfield
  {journal} {\bibinfo  {journal} {IEEE Transactions on Information Theory}\
  }\textbf {\bibinfo {volume} {62}},\ \bibinfo {pages} {1758} (\bibinfo {year}
  {2016})}\BibitemShut {NoStop}%
\bibitem [{Note2()}]{Note2}%
  \BibitemOpen
  \bibinfo {note} {Non-integer $\protect \sqrt {M}$ can easily be taken care of
  as in~\cite {BBMW_full16}.}\BibitemShut {Stop}%
\bibitem [{\citenamefont {Wilde}(2016)}]{W16}%
  \BibitemOpen
  \bibfield  {author} {\bibinfo {author} {\bibfnamefont {M.~M.}\ \bibnamefont
  {Wilde}},\ }\href@noop {} {\  (\bibinfo {year} {2016})},\ \bibinfo {note}
  {arXiv:1106.1445v7}\BibitemShut {NoStop}%
\bibitem [{\citenamefont {Berta}\ \emph {et~al.}(2015)\citenamefont {Berta},
  \citenamefont {Seshadreesan},\ and\ \citenamefont {Wilde}}]{BSW14}%
  \BibitemOpen
  \bibfield  {author} {\bibinfo {author} {\bibfnamefont {M.}~\bibnamefont
  {Berta}}, \bibinfo {author} {\bibfnamefont {K.}~\bibnamefont {Seshadreesan}},
  \ and\ \bibinfo {author} {\bibfnamefont {M.~M.}\ \bibnamefont {Wilde}},\
  }\href@noop {} {\bibfield  {journal} {\bibinfo  {journal} {Journal of
  Mathematical Physics}\ }\textbf {\bibinfo {volume} {56}},\ \bibinfo {pages}
  {022205} (\bibinfo {year} {2015})}\BibitemShut {NoStop}%
\bibitem [{\citenamefont {Ye}\ \emph {et~al.}(2008)\citenamefont {Ye},
  \citenamefont {Bai},\ and\ \citenamefont {Wang}}]{PhysRevA.78.030302}%
  \BibitemOpen
  \bibfield  {author} {\bibinfo {author} {\bibfnamefont {M.-Y.}\ \bibnamefont
  {Ye}}, \bibinfo {author} {\bibfnamefont {Y.-K.}\ \bibnamefont {Bai}}, \ and\
  \bibinfo {author} {\bibfnamefont {Z.~D.}\ \bibnamefont {Wang}},\ }\href
  {\doibase 10.1103/PhysRevA.78.030302} {\bibfield  {journal} {\bibinfo
  {journal} {Physical Review A}\ }\textbf {\bibinfo {volume} {78}},\ \bibinfo
  {pages} {030302} (\bibinfo {year} {2008})}\BibitemShut {NoStop}%
\bibitem [{\citenamefont {Berta}\ \emph {et~al.}(2011)\citenamefont {Berta},
  \citenamefont {Christandl},\ and\ \citenamefont
  {Renner}}]{bertachristandl11}%
  \BibitemOpen
  \bibfield  {author} {\bibinfo {author} {\bibfnamefont {M.}~\bibnamefont
  {Berta}}, \bibinfo {author} {\bibfnamefont {M.}~\bibnamefont {Christandl}}, \
  and\ \bibinfo {author} {\bibfnamefont {R.}~\bibnamefont {Renner}},\ }\href
  {\doibase 10.1007/s00220-011-1309-7} {\bibfield  {journal} {\bibinfo
  {journal} {Communications in Mathematical Physics}\ }\textbf {\bibinfo
  {volume} {306}},\ \bibinfo {pages} {579} (\bibinfo {year}
  {2011})}\BibitemShut {NoStop}%
\bibitem [{\citenamefont {Kim}(2013)}]{kim13}%
  \BibitemOpen
  \bibfield  {author} {\bibinfo {author} {\bibfnamefont {I.~H.}\ \bibnamefont
  {Kim}},\ }\emph {\bibinfo {title} {{Conditional independence in quantum
  many-body systems}}},\ \href@noop {} {Ph.D. thesis},\ \bibinfo  {school}
  {California Institute of Technology} (\bibinfo {year} {2013})\BibitemShut
  {NoStop}%
\bibitem [{\citenamefont {Hayashi}(2015)}]{Hayashi2015}%
  \BibitemOpen
  \bibfield  {author} {\bibinfo {author} {\bibfnamefont {M.}~\bibnamefont
  {Hayashi}},\ }\href@noop {} {\bibfield  {journal} {\bibinfo  {journal} {IEEE
  Transactions on Information Theory}\ }\textbf {\bibinfo {volume} {61}},\
  \bibinfo {pages} {5595} (\bibinfo {year} {2015})}\BibitemShut {NoStop}%
\bibitem [{\citenamefont {Hayashi}(2017)}]{Hayashi2017}%
  \BibitemOpen
  \bibfield  {author} {\bibinfo {author} {\bibfnamefont {M.}~\bibnamefont
  {Hayashi}},\ }\href@noop {} {\emph {\bibinfo {title} {Quantum Information
  Theory - Mathematical Foundation}}}\ (\bibinfo  {publisher} {Springer},\
  \bibinfo {year} {2017})\BibitemShut {NoStop}%
\bibitem [{\citenamefont {Christandl}\ and\ \citenamefont
  {Winter}(2004)}]{CW04}%
  \BibitemOpen
  \bibfield  {author} {\bibinfo {author} {\bibfnamefont {M.}~\bibnamefont
  {Christandl}}\ and\ \bibinfo {author} {\bibfnamefont {A.}~\bibnamefont
  {Winter}},\ }\href@noop {} {\bibfield  {journal} {\bibinfo  {journal}
  {Journal of Mathematical Physics}\ }\textbf {\bibinfo {volume} {45}},\
  \bibinfo {pages} {829} (\bibinfo {year} {2004})}\BibitemShut {NoStop}%
\bibitem [{\citenamefont {Ollivier}\ and\ \citenamefont
  {Zurek}(2001)}]{zurek01}%
  \BibitemOpen
  \bibfield  {author} {\bibinfo {author} {\bibfnamefont {H.}~\bibnamefont
  {Ollivier}}\ and\ \bibinfo {author} {\bibfnamefont {W.~H.}\ \bibnamefont
  {Zurek}},\ }\href {\doibase 10.1103/PhysRevLett.88.017901} {\bibfield
  {journal} {\bibinfo  {journal} {Physical Review Letters}\ }\textbf {\bibinfo
  {volume} {88}},\ \bibinfo {pages} {017901} (\bibinfo {year}
  {2001})}\BibitemShut {NoStop}%
\end{thebibliography}%

\end{document}